\documentclass[a4paper,UKenglish,runningheads,11pt]{llncs}

\usepackage{tabularx,booktabs,multirow,delarray,array}
\usepackage{graphicx,amssymb,amsmath,amsfonts}
\usepackage{enumerate}
\usepackage{latexsym}
\usepackage{enumerate}
\usepackage{complexity}
\usepackage{lineno}
\usepackage{wrapfig}
\usepackage{xcolor}
\usepackage{hyperref}
\usepackage[linesnumbered,ruled,vlined]{algorithm2e}
\usepackage{fullpage}
\usepackage{subcaption}
\usepackage{textgreek}
\usepackage{xcolor}
\usepackage[font={normal}]{caption}
\newtheorem{propo}{Proposition}

\newtheorem{observation}{Observation}
\newenvironment{proof}{\par\noindent{\bf Proof:}}{\mbox{}\hfill$\qed$\\}

\newcommand{\ignore}[1]{ }

\newcounter{rem}
\begin{document}

\title{Vertex Guarding for Dynamic Orthogonal Art Galleries}
\titlerunning{Vertex Guarding for Dynamic Orthogonal Art Galleries}

\author{
Debangshu Banerjee\inst{1}
\and
R. Inkulu\inst{1}
}

\authorrunning{D. Banerjee, R. Inkulu}

\institute{
Department of Computer Science and Engineering\\
IIT Guwahati, India\\
\email{\{d.banerjee,rinkulu\}@iitg.ac.in}
}

\maketitle

\pagenumbering{arabic}
\setcounter{page}{1}

\begin{abstract}
We devise an algorithm for surveying a dynamic orthogonal polygonal domain by placing one guard at each vertex in a subset of its vertices, i.e., whenever an orthogonal polygonal domain $\cal P'$ is modified to result in another orthogonal polygonal domain $\cal P$, our algorithm updates the set of vertex guards surveying $\cal P'$ so that the updated guard set surveys $\cal P$.
Our algorithm modifies the guard placement in $O(k \lg{(n+n')})$ amortized time, while ensuring the updated orthogonal polygonal domain with $h$ holes and $n$ vertices is guarded using at most $\lfloor (n+2h)/4 \rfloor$ vertex guards.
For the special case of the initial orthogonal polygon being hole-free and each update resulting in a hole-free orthogonal polygon, our guard update algorithm takes $O(k\lg{(n+n')})$ worst-case time.
Here, $n'$ and $n$ are the number of vertices of the orthogonal polygon before and after the update, respectively; and, $k$ is the sum of $|n - n'|$ and the number of updates to a few structures maintained by our algorithm.
Further, by giving a construction, we show it suffices for the algorithm to consider only the case in which the parity of the number of reflex vertices of both $\cal P'$ and $\cal P$ are equal.
\end{abstract}

\section{Introduction}

A {\it simple polygon} $T$ is a simply-connected closed region bounded by a finite set of line segments, called {\it edges} of $T$, such that no two edges of $T$ intersect except at their endpoints.
Any endpoint of an edge of $T$ is called a {\it vertex} of $T$.
An {\it orthogonal (simple) polygon} is a simple polygon in which each edge is either parallel to $x$-axis (horizontal) or parallel to $y$-axis (vertical). 
Thus the edges alternate between horizontal and vertical, and always meet orthogonally, with each internal angle being either ${\pi}/{2}$ or ${3\pi}/{2}$. 
The orthogonal polygons are useful as approximations to simple polygons, and they arise naturally in several domains, ex., computer graphics, VLSI design, and computer architecture.
A simple polygon (resp. orthogonal polygon) containing $h \ge 0$ number of disjoint simple polygons (resp. orthogonal polygons) within it is called a {\it polygonal domain (resp. orthogonal polygonal domain)} $\cal T$.
These disjoint simple polygons (resp. orthogonal polygons) are the {\it holes} or {\it obstacles} of $\cal T$.
The {\it free space} $\mathcal{F(T)}$ of the given polygonal domain ${\cal T}$ is the closure of the outer simple polygon of $\cal T$ excluding the union of the interior of holes contained in it. 
A simple polygon (resp. a hole-free orthogonal polygon) is a polygonal domain (resp. an orthogonal polygonal domain) with no holes. 
The {\it reflex parity} of any orthogonal polygonal domain is the parity of the number of reflex vertices defining it. 
In this paper, all the polygons are assumed to be in the plane.
For convenience, we call a hole-free orthogonal polygon as an orthogonal polygon.

Consider a polygonal domain $\cal T$ with its free space $\mathcal{F(T)}$.
Any two points $p', p''$ in $\cal{F(T)}$  are said to be {\it visible} to each other whenever the line segment joining $p'$ and $p''$ lies in $\cal{F(T)}$.
The {\it art gallery problem} seeks to locate a set $G$ of {\it guards} in $\cal{F(T)}$ such that (i) every point $p$ in $\cal{F(T)}$ is visible to at least one guard in $G$, and (ii) the cardinality of $G$ is the minimum possible.
A guard $g$ in $G$ is a {\it vertex guard} if $g$ is located at a vertex of the polygon; otherwise, $g$ is a {\it point guard}.
If all the guards in $G$ are vertex guards, then that polygonal domain is said to be guarded with vertex guards.

In the description below, the number of vertices of a polygon is denoted by $n$.
In \cite{journals/combtheory/Chvat75}, Chvatal had shown that $\lfloor n/3 \rfloor$ vertex guards are both necessary and sufficient to guard a simple polygon.
Fisk gave a different and simpler sufficiency proof of the same result in \cite{journals/combtheory/Fisk78}.
Based on the proof given in \cite{journals/combtheory/Fisk78}, Avis and Toussaint~\cite{journals/pattrecog/AvisTouss81} devised an algorithm to position vertex guards in any simple polygon in $O(n \lg{n})$ time.
Using $O(n)$ time triangulation algorithm by Chazelle \cite{journals/dcg/Chazelle91}, this algorithm essentially takes $O(n)$ time.

In \cite{journals/siamjadm/KahnKK83}, Kahn et~al. gave the art gallery theorem for hole-free orthogonal polygons: they showed $\lfloor n/4 \rfloor$ vertex guards are occasionally necessary and always sufficient to guard any hole-free orthogonal polygon with $n$ vertices.
Their algorithm computes a convex quadrilateralization of the given hole-free orthogonal polygon and places one guard at a vertex of each of these quadrilaterals. 
A convex quadrilateralization of a hole-free orthogonal polygon can be computed using algorithms of Edelsbrunner et~al.~\cite{EdelsRW84}, Lubiw~\cite{conf/scg/Lubiw85}, or Sack and Toussaint~\cite{journals/compmorph/SackTouss88}.
Using $O(n)$ time triangulation algorithm by Chazelle~\cite{journals/dcg/Chazelle91}, the algorithm in \cite{journals/siamjadm/KahnKK83} takes $O(n)$ time.
In \cite{journals/jourgeom/ORourke83}, O'Rourke gave a different approach to show $\lfloor n/4 \rfloor$ vertex guards are sufficient to guard any hole-free orthogonal polygon with $n$ vertices.
In $O(n)$ time, the algorithm in \cite{journals/jourgeom/ORourke83} partitions the given hole-free orthogonal polygon into $L$-shaped orthogonal polygons, and then places one vertex guard at the reflex vertex of every $L$-shaped orthogonal polygon.
(The main ideas of this algorithm are presented in Section~\ref{sect:holefree}.)
The problem of finding a minimum number of vertex guards to guard the given hole-free orthogonal polygon is known to be NP-hard (refer \cite{books/autorobotveh/LeeLin90,jorunals/mlq/SchuHeck95}), and there are a number of approximation algorithms devised for this problem (refer \cite{journals/dam/BhattGhoshRoy17,journals/dam/Ghosh10,journals/cgta/KatzRoi08,journals/cgta/King13}).

In the description below, the number of holes in a polygon with holes is denoted by $h$, and the number of vertices defining the polygon with holes (i.e., the number of vertices of the outer polygon added with the number of vertices of all the holes) is denoted by $n$. 
By applying the convex quadrilaterlization algorithm in \cite{conf/scg/Lubiw85} to an orthogonal polygon with holes, O'Rourke~\cite{books/agtalgo/ORourke87} had shown that $\lfloor (n + 2h)/4 \rfloor$ vertex guards are sufficient to guard the free space of any orthogonal polygon with holes.
Nevertheless, no examples of orthogonal polygons with holes are known to require more than $\lfloor n/4 \rfloor$ vertex guards.
Aggarwal~\cite{books/jhu/Aggar84} established that $\lfloor n/4 \rfloor$ vertex guards suffice when $h = 1$ or $h = 2$. 
Aggarwal~\cite{books/jhu/Aggar84} and Shermer~\cite{tr/nyit/Sher85} respectively conjectured that $\lfloor 3n/11 \rfloor$ and $\lfloor (n+h)/4 \rfloor$ vertex guards are sufficient to guard any orthogonal polygon with $h$ number of holes and $n$ number of vertices defining it.
Hoffmann~\cite{conf/icalp/Hoff90} shown that $\lfloor n/4 \rfloor$ point guards are always sufficient and sometimes necessary to survey any orthogonal polygonal domain.
In addition, in \cite{conf/focs/HoffKK91}, Hoffmann et~al. proved that $\lfloor (n+h)/3 \rfloor$ point guards are sufficient to guard the free space of any given polygonal domain.
To guard the free space of any given polygonal domain with $\lfloor (n+h)/3 \rfloor$ point guards, Bjorling-Sachs and Souvaine~\cite{journals/dcg/BjorsachsSouv95} devised an $O(n^2)$ time algorithm.

Several algorithms for visibility computations are detailed in the text by Ghosh in \cite{books/visalgo/skghosh2007}.
The text by O'~Rourke \cite{books/agtalgo/ORourke87} gives a detailed presentation of various well-known algorithms for guarding art galleries.
Algorithms for a number of guarding and illumination problems are presented in \cite{incoll/cghb/Urru00}.

When a polygonal domain is modified, instead of applying any of the algorithms mentioned above, it is desirable to have a {\it local update algorithm}, which essentially changes the set of guards and their positions locally with respect to the recent modification of the polygonal domain.
A polygonal domain (resp. orthogonal polygonal domain) is called a {\it dynamic polygonal domain} (resp. {\it dynamic orthogonal polygonal domain}) whenever a specific set of discrete events modify it while each such modification yields a polygonal domain (resp. an orthogonal polygonal domain).

For any point $p$ located in the free space $\cal{F(T)}$ of a polygonal domain $\cal T$, the {\it visibility polygon} $VP_{\cal T}(p)$ of $p$ is the maximal set $S$ of points in $\cal{F(T)}$ such that every point in $S$ is visible to $p$.
The {\it visibility polygon query} problem seeks to preprocess the given polygonal domain $\cal T$ so that to efficiently compute the visibility polygon of any point $p \in \cal{F(T)}$ given in the query phase of the algorithm. 
Inkulu and Nitish \cite{conf/caldam/Inkulu17} and Inkulu et~al. \cite{journals/ijcga/InkuluST20} devised algorithms for maintaining the visibility polygon of any query point in dynamic simple polygons.
For the case of dynamic polygonal domains, recently Agrawal and Inkulu~\cite{conf/cocoon/AgrawInkulu20} presented algorithms for maintaining the visibility polygon of any query point.
The {\it (vertex-vertex) visibility graph} of a polygonal domain $\cal T$ is the undirected graph with its vertex set comprising all the vertices of $\cal T$ and the edge set comprising every line segment with its endpoints $v', v''$ being the vertices of $\cal T$ such that $v'$ and $v''$ are visible.
For dynamically maintaining the visibility graph, Choudhury and Inkulu~\cite{conf/caldam/Inkulu19} and Agrawal and Inkulu~\cite{conf/cocoon/AgrawInkulu20} presented algorithms.
However, to our knowledge, there are no algorithms known for updating the guard set in dynamic domains.

\subsection*{Our results}

In this paper, for dynamic orthogonal polygonal domains, an algorithm is proposed to update the set of vertex guards and to reposition a subset of the vertex guards currently surveying, as the polygonal domain is dynamically modified. 
Our algorithm uses at most $\lfloor (n+2h)/4 \rfloor$ vertex guards when the updated polygonal domain has $n$ vertices and $h$ orthogonal holes.
Our algorithm modifies the guard placement local to where the orthogonal polygonal domain is modified.
As part of this, by using thin horizontal rectangles, we convert any given orthogonal polygon with holes into a hole-free orthogonal polygon; these rectangles are called {\it channels}.
The resultant hole-free orthogonal polygon $P$ is partitioned into L-shaped orthogonal polygons by horizontally/vertically projecting a subset of reflex vertices of $P$; the line segments resultant from these projections are called {\it cuts}.
(The precise definitions of cuts and channels are given in later sections.)
Our algorithm takes $O(k \lg{(n +  n')})$ amortized time to update the set of vertex guards whenever the current orthogonal polygonal domain $\cal P'$ is updated.
Here, $n'$ and $n$ are the number of vertices of the orthogonal polygonal domain before and after the update, respectively; $k$ is the sum of the number of vertices added to or deleted from $\cal P'$, the number of cuts in the L-shaped partitioning of $\cal{F(P')}$ that got affected due to the modification, and the number of channels in $\cal{F(P')}$ that got affected due to the modification.
When there are no holes in the initial input orthogonal polygon as well as in subsequent orthogonal polygons resultant of any update, our guard update algorithm takes $O(k \lg{(n +  n')})$ worst-case time.
The initial orthogonal polygonal domain $\cal Q$ with $q$ vertices (the polygon before any dynamic updates) is preprocessed in $O(q \lg{q})$ time to construct a few data structures of size $O(q\frac{\lg{q}}{\lg\lg{q}})$ and to survey $\cal Q$ with $\lfloor q/4 \rfloor$ vertex guards.
Further, we give a construction to show that any vertex guarding algorithm to handle dynamic updates requires to consider only the case in which the parity of the number of reflex vertices of both $\cal P'$ and $\cal P$ are equal; otherwise, we show all the vertex guards may need to be re-positioned.

To our knowledge, this is the first algorithm for guarding a dynamic art gallery.
This algorithm obviates to re-compute the entire set of guards to survey the modified polygonal domain whenever a small section of the polygonal domain is modified.
Though the $k$ in the time complexities of algorithms for both the orthogonal polygon with holes as well as without holes is $O(n+n')$, the value of $k$ is in general much smaller to $n$, since the dynamic updates are typically local. 
Hence, our algorithm is in general efficient in handling dynamic updates as compared to applying directly the traditional vertex guarding algorithms to the updated orthogonal polygonal domain.

\subsubsection*{Our approach}

Here we give an outline of our approach.
We reduce the problem of vertex guarding any orthogonal polygonal domain with $h$ holes and $n$ vertices to the problem of vertex guarding a hole-free orthogonal polygon with $(n + 2h)$ vertices.
This is accomplished by constructing thin horizontal rectangles (channels) in the free space of the orthogonal polygon with holes.
We specialize the channel notion in Bjorling-Sachs and Souvaine~\cite{journals/dcg/BjorsachsSouv95} to orthogonal polygons with orthogonal holes.
Mainly, after every update to the orthogonal polygon with holes, we update the relevant channels to transform the updated orthogonal polygonal domain $\cal P$ into a hole-free orthogonal polygon $P$. 
Then we update the set of vertex guards to guard the hole-free orthogonal polygon using the algorithm for dynamic hole-free orthogonal polygons.
We identify a minimal sized orthogonal polygon $R \subset P$ for which vertex guards need to be determined afresh. 
We independently vertex guard $R$ using the algorithm in \cite{journals/jourgeom/ORourke83}. 
The vertex guards computed for $R$ and the vertex guards located in $P-R$ together are shown to guard $P$, and in turn $\cal P$.
We maintain the hole-free orthogonal polygon with $n''$ vertices corresponding to any orthogonal polygonal domain with $h$ holes and $n$ vertices, such that $n'' = (n+2h)$.
Notably, if $n''$ is the number of vertices of the hole-free orthogonal polygon $P$, our algorithm places at most $\lfloor n''/4 \rfloor$ vertex guards to guard $P$.

Before any modification to the initial input orthogonal polygonal domain $\cal Q$ (the one before any dynamic updates) defined with $h$ holes and $q$ vertices, by removing channels (a set of thin horizontal rectangles) from $\cal Q$, we compute a hole-free orthogonal polygon $Q$.
Then using the algorithm in \cite{journals/jourgeom/ORourke83}, we partition $Q$ into $\lfloor (q+2h)/4 \rfloor$ L-shaped orthogonal polygons, and guard each such L-shaped orthogonal polygon with one vertex guard.
In addition, for efficiently updating channels and cuts, we construct a few data structures in the preprocessing phase.

\subsubsection*{Terminology}

We assume the initial orthogonal polygonal domain and every orthogonal polygonal domain that is resulted due to updates are in general position, i.e., there are no two reflex vertices visible to each other either along a vertical line segment or along a horizontal line segment.
A hole-free orthogonal polygon is also called a {\it piece}.
The initial input hole-free orthogonal polygon (resp. orthogonal polygonal domain), the one before any of updates, is denoted by $Q$ (resp. $\cal Q$).
The hole-free orthogonal polygon (resp. orthogonal polygonal domain) just before any update is denoted by $P'$ (resp. $\cal P'$).
And, the hole-free orthogonal polygon (resp. orthogonal polygonal domain) just after any update is denoted by $P$ (resp. $\cal P$).
For any hole-free orthogonal polygon $P$, the boundary of $P$ is denoted by $bd(P)$.

\medskip

Section~\ref{sect:holefree} details an algorithm for vertex guarding dynamic hole-free orthogonal polygons.
The algorithm to update the set of vertex guards of a dynamic orthogonal polygonal domain is presented then in Section~\ref{sect:holes}.
The conclusions are in Section~\ref{sect:conclu}.

\section{Handling updates in a dynamic hole-free orthogonal polygon}
\label{sect:holefree}

In this section, we devise an algorithm to update the set of vertex guards when the initial input is a hole-free orthogonal polygon, and each update also leads to a hole-free orthogonal polygon.
The initial input orthogonal polygon $Q$ with $q$ vertices is preprocessed to compute a few data structures, to partition $Q$ into L-shaped pieces, and in turn, for guarding $Q$ using at most $\lfloor q/4 \rfloor$ vertex guards. 
In the update algorithm, we separate an orthogonal polygon $R$ from the updated orthogonal polygon $P$ wherein the vertex guarding of $R$ may require modifying, and hence we independently determine vertex guards to survey $R$.
In Subsection~\ref{subsect:charaffreg}, we define a few properties of $R$, and the algorithm to separate $R$ from $P$ is given in Subsection~\ref{subsect:sepaffreg}.

\subsection{Characterizing the affected region of $P$}
\label{subsect:charaffreg}

We first briefly present a few observations from \cite{journals/jourgeom/ORourke83}.
It is known that $\lfloor n/4 \rfloor$ vertex guards are occasionally necessary and always sufficient to guard any orthogonal polygon with $n$ vertices.
Since a single vertex guard can guard any L-shaped orthogonal polygon, to guard any hole-free orthogonal polygon $T$ with $n$ vertices, it suffices to partition $T$ into at most $\lfloor n/4 \rfloor$ L-shaped orthogonal polygons. 
Let $r$ be the number of reflex vertices of $T$.
Then, it is immediate to note that $(n - 2)\pi = (r\frac{3\pi}{2}) + ((n-r)\frac{\pi}{2})$.
Hence, $n = 2r + 4$, and so $\lfloor n/4 \rfloor = \lfloor r/2 \rfloor + 1$. 
Therefore, to survey a hole-free orthogonal polygon with $n$ vertices using at most $\lfloor n/4 \rfloor$ vertex guards, it is sufficient to partition $T$ into $\lfloor r/2 \rfloor + 1$ L-shaped pieces. 
In any orthogonal polygon $T$, for any reflex vertex $v$, the {\it horizontal cut} (resp. {\it vertical cut}) incident to $v$ is the horizontal (resp. vertical) line segment joining $v$ with a point $p$ on $bd(T)$ such that the open line segment $pv$ is located interiorly to $T$.
A horizontal or vertical cut $C$ that is incident to a reflex vertex $v$ resolves $v$, i.e., vertex $v$ is no longer reflex in either of the two pieces of the partition determined by $C$.
A cut $C$ is said to be an {\it odd cut} if at least one of the two pieces determined by $C$ has an odd number of reflex vertices. 
The main observation in \cite{journals/jourgeom/ORourke83} is that odd cuts help in devising a natural divide-and-conquer algorithm to partition the orthogonal polygon into L-shaped pieces: If we partition an orthogonal polygon by cutting along any odd cut $C$, the two pieces that are determined by $C$ can be partitioned into L-shaped pieces independently.
Let $C$ be an odd cut of a hole-free orthogonal polygon $T$ (such that $C$ partitions $T$ into two pieces).
Also, let $r_1$ and $r_2$ be the number of reflex vertices in each of the pieces determined by $C$, without loss of generality, say, $r_1$ is odd.
Then, as $r = r_1 + r_2 + 1$, and since $r_1$ is odd, it is immediate to note that $\lfloor r/2 \rfloor + 1 = \lfloor (r_1 - 1)/2 \rfloor + 1 +  \lfloor r_2/2 \rfloor + 1 = \lfloor r_1/2 \rfloor + 1 +  \lfloor r_2/2 \rfloor + 1$.
This says that if the piece with $r_1$ reflex vertices can be partitioned into at most $\lfloor r_1/2 \rfloor + 1$ L-shaped pieces, and the piece with $r_2$ reflex vertices can be partitioned into at most $\lfloor r_2/2 \rfloor + 1$ L-shaped pieces, then $T$ can be partitioned into at most $\lfloor r/2 \rfloor + 1$ L-shaped pieces, provided $T$ has an odd cut.
However, as shown in \cite{journals/jourgeom/ORourke83}, there exists an odd cut in any orthogonal polygon that is in general position.
Therefore, it is evident that to partition any orthogonal polygon into L-shaped pieces, one needs to find odd cuts efficiently.

\begin{figure}[h]
    \centering
    \begin{subfigure}{.3\textwidth}
    \centering
    \includegraphics[height=2cm]{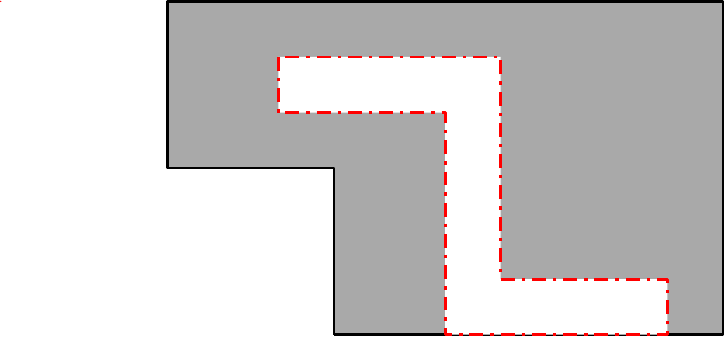}
    \end{subfigure}
    \qquad
    \begin{subfigure}{.3\textwidth}
    \centering
    \includegraphics[height=2cm]{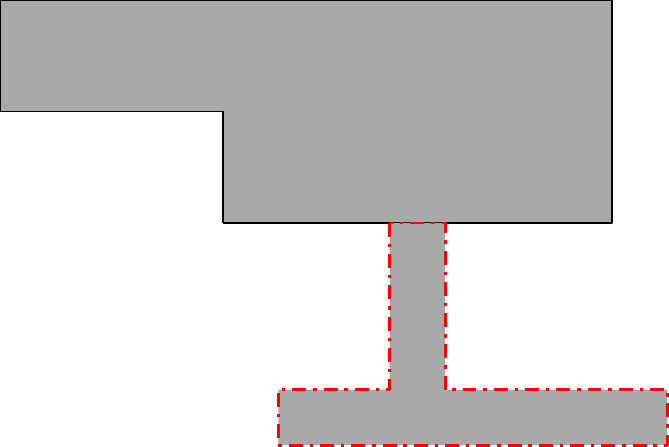}
    \end{subfigure}
    \centering
    \begin{subfigure}{.3\textwidth}
    \centering
    \includegraphics[height=2cm]{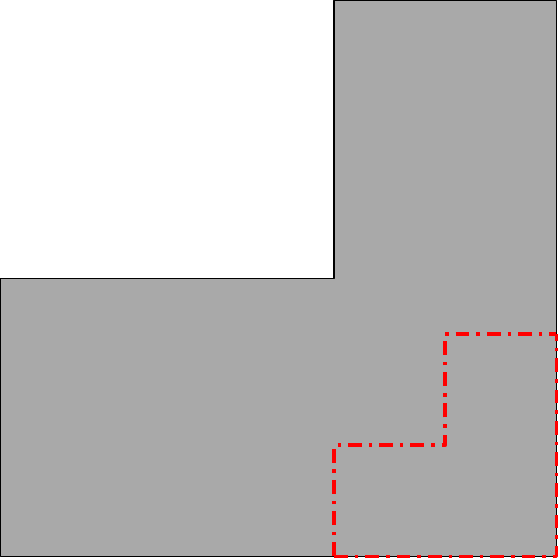}
    \qquad
    \end{subfigure}
    \caption{
    The left figure illustrates a type-I update in which $R'$ is subtracted from $P'$. 
    The middle and right figures show a type-II update; in these, $R'$ is unioned with $P'$, while $bd(R') \cap bd(P')$ is a staircase.
    The boundary of $R'$ is shown in red color.
    The updated hole-free orthogonal polygon $P$ is shaded.
    }
    \label{fig:typeIupdates}
\end{figure}
We preprocess $Q$ using the algorithm in \cite{journals/jourgeom/ORourke83}.
(To remind, $Q$ is the orthogonal polygon before any of updates.) 
The algorithm in \cite{journals/jourgeom/ORourke83} partitions $Q$ into L-shaped pieces, and guards each such L-shaped piece with one vertex guard.
Let $P'$ be an orthogonal polygon just before any update.
We support two types of updates to $P'$ and both of these updates ensure the updated polygon $P$ is a hole-free orthogonal polygon and the reflex parity of $P$ is same as the reflex parity of $P'$.
In a {\it type-I update}, an orthogonal polygon $R'$ is subtracted from $P'$ while $R'$ is positioned such that (a) $R'$ is interior to $P'$, and (b) an edge of $R'$ is abutting a section of an edge of $P'$.
In a {\it type-II update}, an orthogonal polygon $R'$ is unioned with $P'$ while $R'$ is positioned such that (a) $R'$ is exterior to $P'$, and (b) $R'$ abuts $P'$ along a rectilinear staircase on $bd(R')$.
Since a line segment is a staircase, uniting $R'$ with $P'$ along an edge of $P'$ is a type-II update.
(Refer to Fig.~\ref{fig:typeIupdates}.)
It is immediate to note that a type-I update suffice to delete any hole-free orthogonal polygon from $P'$, and a type-II update suffice to unite any hole-free orthogonal polygon with $P'$ along a staircase on $bd(P')$, except for the reflex parity restriction imposed on $P$ in both of these updates.

In the following, by giving an example construction, we illustrate that without the reflex parity restriction on $P$ in these updates, every guard in $P'$ may need to be re-positioned so that to guard $P$ with $n$ vertices using at most $\lfloor n/4 \rfloor$ vertex guards.
The left orthogonal polygon in Fig.~\ref{fig:twoupdates} has $10$ vertices, and it can be guarded using two vertex guards.
By attaching two rectangles in two successive type-II updates, one touching edge $ab$, and the other touching edge $hg$ of $P'$, the resultant polygon after these two updates is the orthogonal polygon on the right in Fig.~\ref{fig:twoupdates}. 
\begin{wrapfigure}{r}{0.5\textwidth}
    \centering
    \begin{minipage}[t]{\linewidth}
    \centering
    \includegraphics[height=3cm]{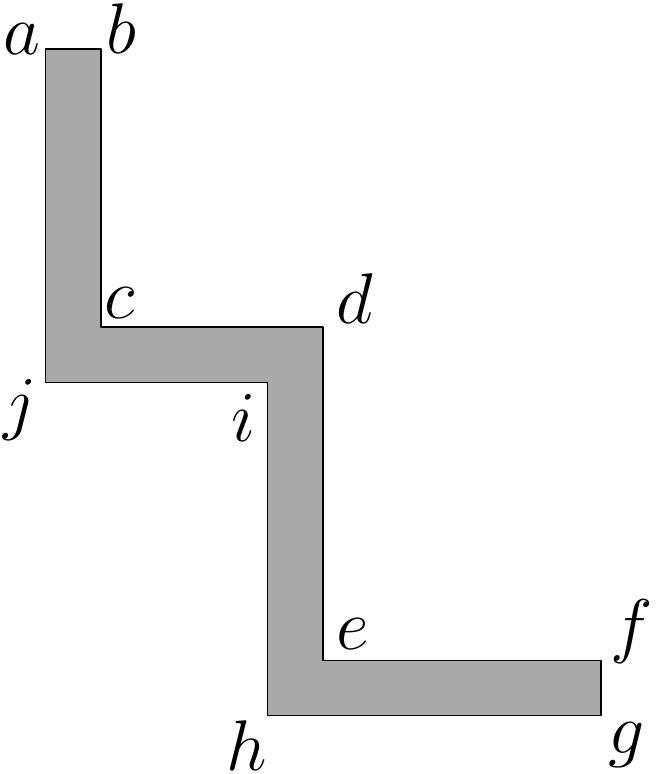}
    \includegraphics[height=3cm]{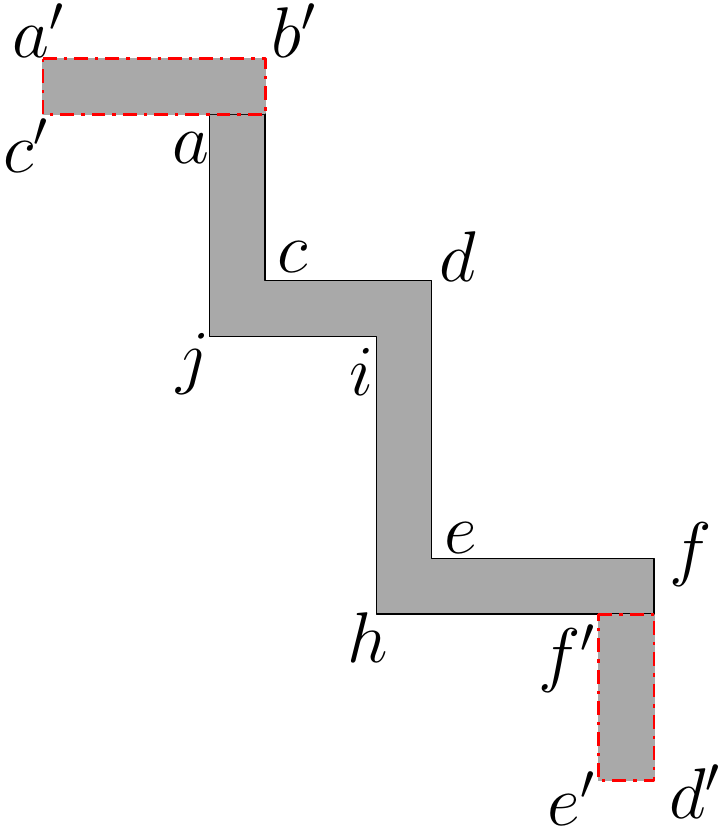}
    \end{minipage}
    \caption{Illustrating two rectangles attached to the polygon on the left resulting in the updated polygon shown on the right.  Specifically, guarding the polygon on the right requires relocating all the vertex guards for the polygon on the left.}
    \label{fig:twoupdates}
\end{wrapfigure}
And, this updated polygon can be guarded using three vertex guards; however, positions of all the vertex guards for surveying the left polygon in Fig.~\ref{fig:twoupdates} need to be re-positioned for guarding the polygon on the right in Fig.~\ref{fig:twoupdates}.
To extend this argument, consider a polygon $P''$ of analogous shape (a contiguous sequence of L-shapes) with an asymptotically large number of reflex vertices, say $\ell$.
When $P''$ is updated as mentioned above, to obey the upper bound on the number of vertex guards, it necessitates relocating \textOmega($\ell$) guards positioned at vertices of $P''$.

When an orthogonal polygon $R'$ is removed from an orthogonal polygon $P'$, resulting in an orthogonal polygon $P$, a cut $C$ in $P'$ is said to be an {\it affected cut} whenever (i) $C$ intersects $R'$ or $C$ abuts an edge of $R'$, or (ii) $C$ is incident to an affected cut.
(Refer to Fig.~\ref{fig:affected_cut}.) 
A cut in $P$ that is not an affected cut is an {\it unaffected cut}.

\begin{wrapfigure}{r}{0.4\textwidth}
    \vspace{-0.4in}
    \centering
    \begin{minipage}[t]{1.2\linewidth}
    \centering
    \includegraphics[width=3cm]{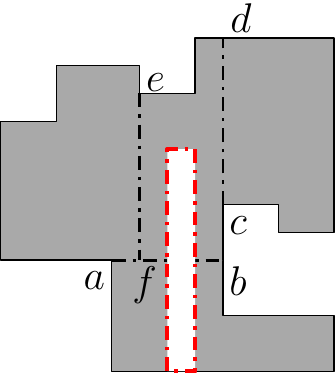}
    \end{minipage}
    \caption{
    Illustrating affected cuts $ab$ and $ef$.
    The cut $ab$ is an affected cut since the line segment $ab$ intersects $R'$.
    The cut $ef$ is an affected cut since it is incident on the affected cut $ab$.
    }       
    \label{fig:affected_cut}
\end{wrapfigure}
Without loss of generality, we assume the polygon $P$ has an odd number of reflex vertices.
Otherwise, as in \cite{journals/jourgeom/ORourke83}, we can introduce an additional reflex vertex by removing a rectangular chip around a convex vertex, and adding that chip back after guarding the rest of the orthogonal polygon.
As detailed in \cite{journals/jourgeom/ORourke83}, the advantage of having an odd number of reflex vertices is that if we split the polygon into two by cutting along an odd cut, then the parity of all the cuts in two smaller polygons remain unchanged.
This fact does not hold for polygons with an even number of reflex vertices. 

We note that for any unaffected horizontal odd cut $C$ in $P$, the newly introduced reflex vertices in $P$ (that are not present in $P'$) are to one of the sides of $C$.
The same holds for vertical odd cuts as well.
Since the reflex parity of both $P'$ and $P$ is odd, all unaffected horizontal odd cuts (resp. vertical odd cuts) remain horizontal odd cuts (resp. vertical odd cuts) after modifying $P'$ with a type-I or a type-II update.

\begin{observation}
If $C$ is an unaffected horizontal odd cut (resp. vertical odd cut) in a hole-free orthogonal polygon $P'$, then $C$ is a horizontal odd cut (resp. vertical odd cut) in $P$.
\end{observation}

In particular, cutting along all unaffected horizontal cuts lead to a set $S$ of polygons, with each polygon in $S$ having only vertical cuts.
Indeed, as noted in \cite{journals/jourgeom/ORourke83}, each polygon in $S$ is a union of two histograms, both having the same horizontal line segment as their base.
Essentially, as all unaffected vertical and horizontal cuts are odd cuts, we separate the {\it affected region} $R$ from $P$ by cutting along a subset of these cuts. 
Here, $R$ is a minimal sized orthogonal polygon that intersects all the affected cuts in $P$.
(Refer to Fig.~\ref{fig:unaffvertcutsvalid}.)
The following observation is helpful for our algorithm.

\begin{wrapfigure}{r}{0.4\textwidth}
    \centering
    \begin{minipage}[t]{1.2\linewidth}
    \centering
    \includegraphics[height=3.8cm]{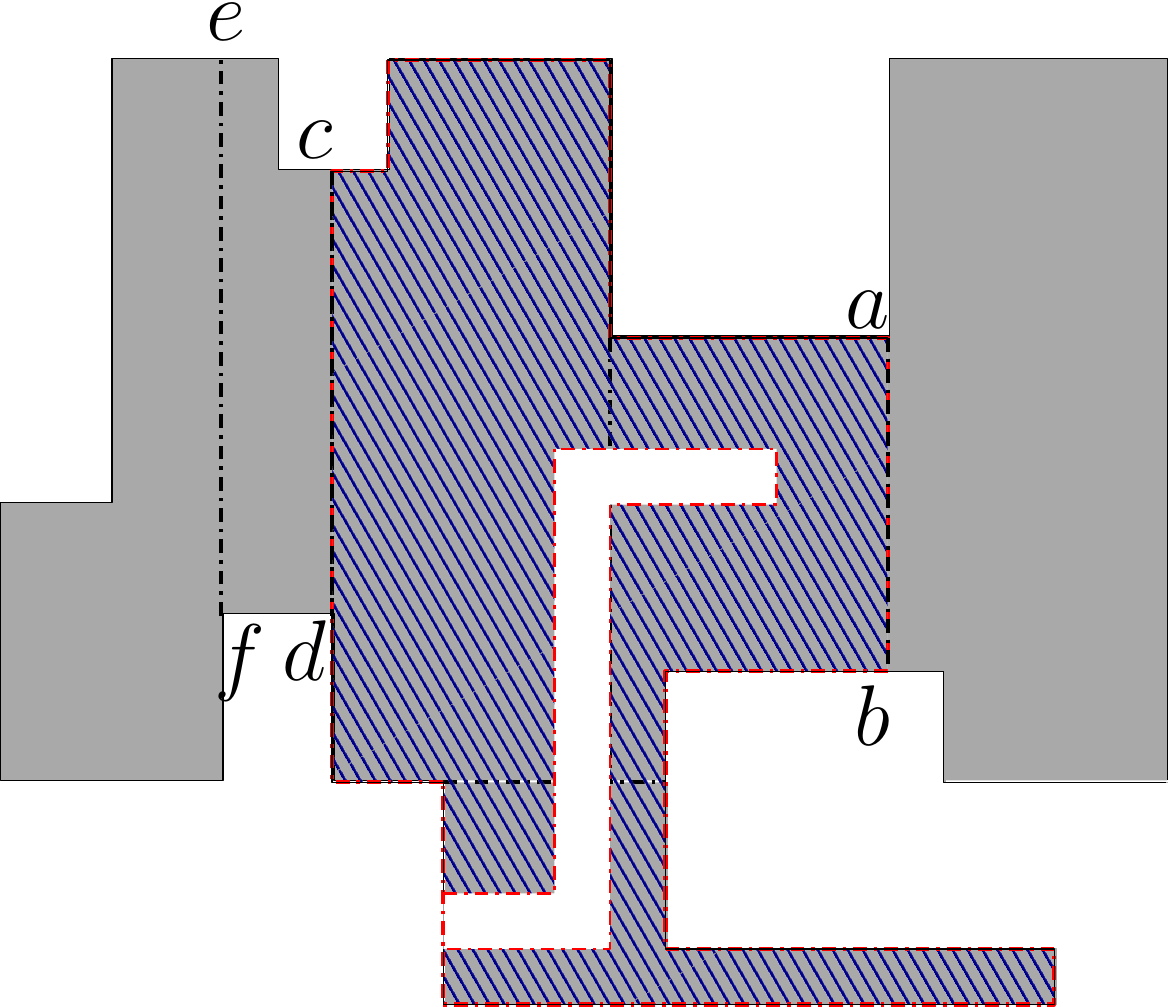}
    \end{minipage}
    \caption{Illustrating unaffected (vertical) cuts $ab$, $cd$, and $ef$ in $P'$, which remain unaffected cuts in $P$ as well.  The shaded region is the polygon $R$.}
    \label{fig:unaffvertcutsvalid}
\end{wrapfigure}

Since $R$ is an orthogonal polygon, for guarding $R$, we use the algorithm in \cite{journals/jourgeom/ORourke83}.
From the following observation, the number of vertex guards used to guard the resulting $n$-vertex polygon $P$ is at most $\lfloor n/4 \rfloor$. 

\begin{observation}
For any hole-free orthogonal polygon $P$ with an odd number of reflex vertices, and any unaffected cut $C$ corresponding to an affected region $R$ of $P$, where $C$ belongs to a connected component $CC$ of $P-R$, the parity of $C$ in $CC$ is same as the partiy of $C$ in $P$.
\end{observation}

\begin{observation}
Let $S'$ (resp. $S''$) be the set comprising vertex guards to guard $R$ (resp. $P'$), determined by applying the algorithm in \cite{journals/jourgeom/ORourke83} to $R$ (resp. $P'$).
Also, let $S'''$ be the set $\{ v \in S'' \hspace{0.02in} | \hspace{0.02in} v$ is located in $P-R\}$.
As every maximal line segment in $bd(R) - bd(P)$ is an odd cut in $P$, the cardinality of $S' \cup S'''$ is upper bounded by $\lfloor n/4 \rfloor$, where $n$ is the number of vertices of $P$.
\end{observation}

In the following section, we describe an algorithm to compute the affected region $R$ efficiently.

\subsection{Algorithm for separating the affected region $R$ from $P$}
\label{subsect:sepaffreg}

To help in finding the affected region $R$ efficiently, in the preprocessing phase, we construct a data structure $\cal S$ comprising all the horizontal cuts in $Q$.
(To remind, $Q$ is the orthogonal polygon before any updates.)
This data structure helps in finding all the horizontal cuts that intersect each vertical edge of $R'$.
For this purpose, we use a data structure from Mortensen~\cite{conf/soda/Morten03}.
Given a vertical ray $r$, using this data structure, it takes $O(k + \lg{q'})$ worst-case time to determine all the horizontal line segments in the data structure that are intersected by $r$.
Here, $q'$ is the number of horizontal line segments present in the data structure at the time of the operation and $k$ is the output size.
This data structure stores any $q'$ number of horizontal line segments using $O(q'\frac{\lg{q'}}{\lg\lg{q'}})$ space; it takes $O(\lg{q'})$ worst-case time to insert any horizontal line segment into data structure or to delete any horizontal line segment from the data structure. 
The data structure $\cal S$ is constructed by inserting all the horizontal cuts in $Q$.
At the end of processing any update to $P'$, our algorithm ensures $\cal S$ precisely comprises horizontal cuts in $P$.
\begin{wrapfigure}{r}{0.42\textwidth}
    \centering
    \vspace{-0.2in}
    \begin{minipage}[t]{\linewidth}
    \centering
    \includegraphics[height=2.8cm]{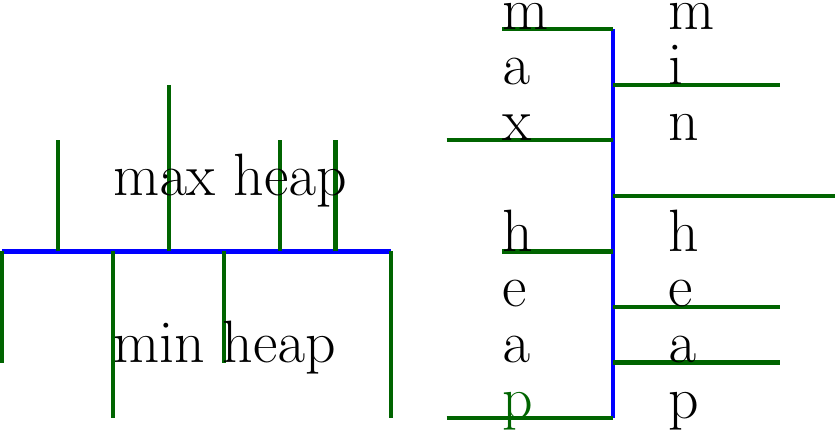}
    \end{minipage}
    \caption{
	Illustrating heaps maintained by horizontal cuts, vertical cuts, horizontal edges, and vertical edges to help in determining the $bd(R)$ by traversing it in clockwise direction. 
	The cut $C$ is shown in blue color.
	}
    \label{fig:heaps}
\end{wrapfigure}
By removing each affected horizontal cut that is reported to have been intersected with a vertical edge of $R'$ from $\cal S$, we ensure that no affected horizontal cut is reported more than once.
Analogously, we find all the vertical cuts that intersect horizontal edges of $R'$ (i.e., by maintaining vertical cuts of $P'$ in a data structure).

We maintain at most two heaps with each cut and each edge to help in efficiently tracing $bd(R)$, specifically, for finding vertices of $bd(R)$ in the clockwise order of their occurrence along $bd(R)$. 
Specifically, we maintain two heaps for each cut in the polygon, one of these is a max heap and the other one is a min heap.
Let $C$ be any horizontal cut.
Also, let $S_C$ be the set comprising all the vertical cuts that are intersecting $C$.
Then, the bottom (resp. top) endpoint of each cut $C' \in S_C$ that is lying above (resp. below) $C$ is stored in a max heap (resp. min heap) associated with $C$.
Analogously, endpoints of cuts that are incident on a vertical cut are distributed into two heaps.
(Refer to Fig.~\ref{fig:heaps}.)
In addition, for each edge $e$ of $Q$, we maintain either a max heap or a min heap with $e$.
For a vertical edge $e$ that is bounding $P$ from the right (resp. left), all the horizontal cuts that intersect $e$ are stored in a max heap (resp. min heap) associated with $e$.
Analogously, each horizontal edge of $P$ is associated with a heap.
For any heap $H$ associated with any horizontal (resp. vertical) cut/edge, for every point $p$ stored in $H$, the $x$-coordinate (resp. $y$-coordinate) of $p$ is the key value of $p$ in $H$.
To remind, max heap (resp. min heap) is a priority queue that supports extracting and querying $H$ for an element in $H$ that has a largest (resp. smallest) key value.

First, we modify heaps by deleting all cuts intersected by $R'$.
The affected region $R$ is found by determining the vertices of $R$. 
Starting from an arbitrary vertex $v$ of $R'$ that is located on an edge $e$ of $P'$, vertices of $R$ are determined in the clockwise order of their occurrence along $bd(R)$.
We first note that each vertex of $R$ is either a vertex of $P'$, or a vertex of $R'$, or an endpoint of an unaffected cut.
Let $v'$ be a vertex of $P'$ that first occurs next to $v$ on $bd(P')$ in the clockwise direction. 
\begin{wrapfigure}{r}{0.42\textwidth}
    \centering
    \vspace{-0.2in}
    \begin{minipage}[t]{\linewidth}
    \centering
    \includegraphics[height=4.4cm]{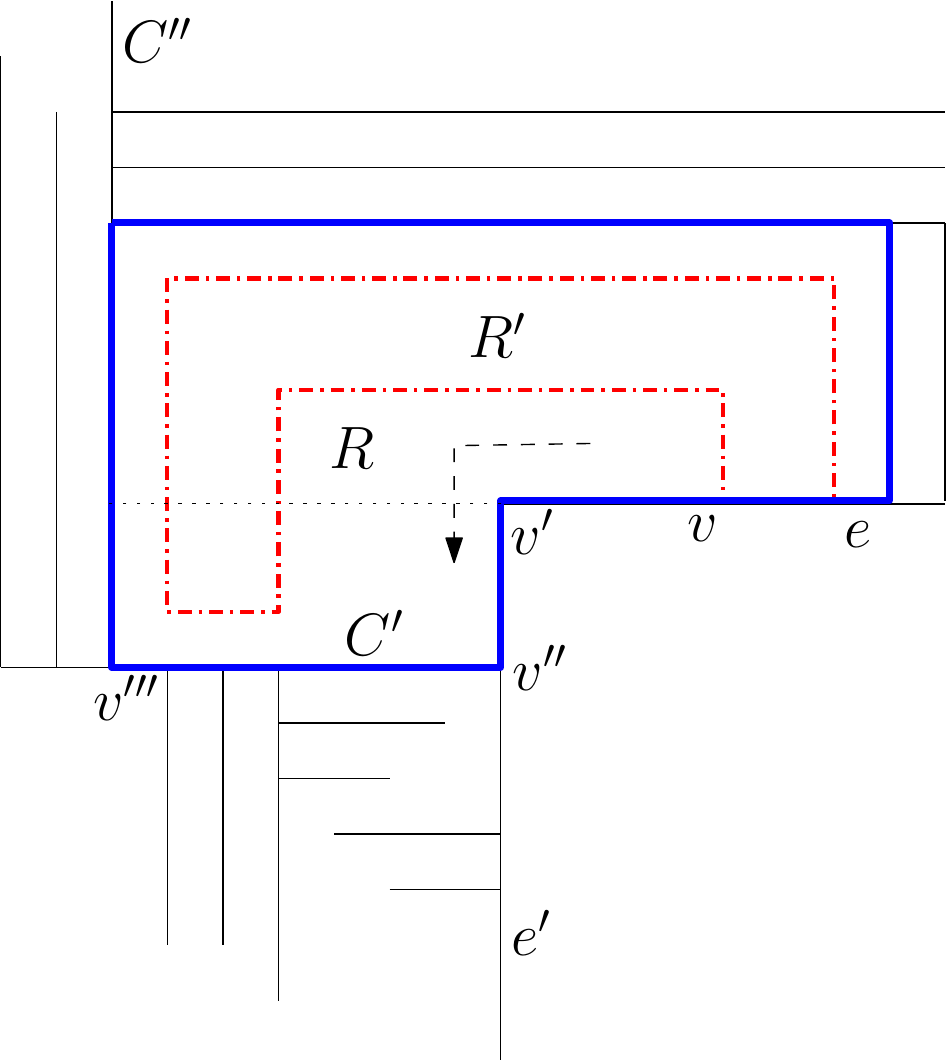}
    \end{minipage}
    \caption{
	    Illustrating the affected region $R$ (blue).
	    Algorithm starts from a vertex $v$ of $R'$, which is located on $bd(P')$.
	    The first vertex of $P'$ while traversing $bd(P')$ from $v$ in clockwise direction is $v'$.
	    Extracting cut with maximum key value from max heap associated with $e'$ yields $v''$. 
	    Again, extracting cut with maximum key value from max heap associated with $C'$ yields $v'''$. 
	    (Note that there are two heaps associated with $C'$.)
	    The dotted line incident to $v'$ is a cut intersected by $R'$.
    }
    \label{fig:traversal}
\end{wrapfigure}
Then $v'$ is a vertex of $R$. 
Also, let $e'$ be the other edge incident to $v'$.
By using the heap associated with $e'$, depending on the heap type, we extract the point with the maximum/minimum key value stored in the heap.
This point (located on $e'$) is the vertex $v''$ of $R$ that occurs next to $v'$ in the clockwise order along $bd(R)$.
In the clockwise traversal of $bd(R)$, if a section of a horizontal cut or a horizontal edge, say $f$, of $R$ is traversed from right to left (resp. left to right), then max heap (resp. min heap) associated with $f$ is used.
And, if a section of a vertical cut or a vertical edge, say $f$, of $P$ is traversed from top to bottom (resp. bottom to top), then max heap (resp. min heap) associated with $f$ is used.
If $C'$ is the cut with endpoint $v''$, then again by using a heap associated with $C'$, we find the vertex that occurs next to $v''$ in the clockwise order along $bd(R)$.
In this fashion, the algorithm continues to find vertices of $R$, until $v'$ is again found by the algorithm. 
(Refer to Fig.~\ref{fig:traversal}.)

\begin{theorem}
\label{thm:orthopolycorr}
Given a hole-free orthogonal polygon $Q$ defined with $q$ vertices, we preprocess $Q$ in $O(q \lg{q})$ time to construct data structures of size $O(q \frac{\lg{q}}{\lg {\lg {q}}})$ so that whenever any orthogonal polygon $P'$, which is obtained by a sequence of type-I and type-II updates to $Q$, is updated to an orthogonal polygon $P$, the algorithm takes $O(k \lg{(n + n')})$ worst-case time to guard $P$ using at most $\lfloor n/4 \rfloor$ vertex guards.
Here, $n'$ and $n$ are the number of vertices of $P'$ and $P$ respectively, and $k$ is the sum of $|n-n'|$ and the number of affected cuts in $P$.
\end{theorem}
\begin{proof}
In Subsection~\ref{subsect:charaffreg}, we argued that all the unaffected odd cuts remain odd cuts.
Following the argument in \cite{journals/jourgeom/ORourke83}, we can always cut any orthogonal polygon at an odd cut and guard the two smaller pieces defined by that cut independently; and, for any orthogonal polygon $T$ in general position, it is guaranteed that there exists an odd cut in $T$.
The affected region $R$ is separated from $P$ by cutting along a set $S'$ of unaffected cuts; these together separate a minimal sized polygon $R$ that contains all the affected cuts in $P$.
The set $S'$ of unaffected cuts and a subset of edges of $P$ together define $bd(R)$.
The orthogonal polygon $P - R$ has L-shaped partitioning, and it is guarded using at most $\lfloor n''/4 \rfloor$ vertex guards as part of guarding $P'$.
Here, $n''$ is the number of vertices of $P - R$.
The orthogonal polygon $R$ is guarded using the algorithm described in \cite{journals/jourgeom/ORourke83}; hence, the number of vertex guards to guard $R$ is upper bounded by $\lfloor r'/4 \rfloor$, where $r'$ is the number of vertices of $R$.
Since every unaffected cut in $bd(R)$ is shown to be an odd cut, $\lfloor n''/4 \rfloor + \lfloor r'/4 \rfloor$ is at most $\lfloor n/4 \rfloor$.

The preprocessing involves L-shaped partitioning using \cite{journals/jourgeom/ORourke83}, constructing the data structure $\cal S$ (following \cite{conf/soda/Morten03}), and initializing heaps. 
Constructing these structures together takes $O(q \lg{q})$ worst-case time, and the size of these structures is $O(q \frac{\lg{q}}{\lg{\lg{q}}})$.
Futher, we note that each cut is placed in two heaps.
Let $k_1$ be the number of affected cuts in $P$.
Also, let $k_2$ be the number of vertices of $R'$. 
It is immediate to note the number of vertices of $R$ is $O(k_1 + k_2)$.
Hence, the number of L-shaped pieces in $R$ is $O(k_1 + k_2)$, i.e., $O(k)$.
Using the query algorithm in \cite{conf/soda/Morten03} for outputting the intersection of dynamic horizontal (resp. vertical) line segments with a vertical (resp. horizontal) line segment, finding all the affected cuts together takes $O((k_1 + k_2) \lg{n'})$ worst-case time.
Deleting affected cuts from the respective heaps takes $O(k_1 \lg{n'})$ worst-case time (by providing a pointer to the node that has the affected cut to the delete method).
The algorithm in \cite{journals/jourgeom/ORourke83} takes $O(k \lg{k})$ time in the worst-case to guard as well as to yield an L-shaped partitioning of $R$. 
From \cite{conf/soda/Morten03}, introducing entries corresponding to $O(k)$ cuts generated by L-shaped partitioning of $R$ into $\cal S$ takes $O(k\lg{n})$ worst-case time.
Updating heaps by inserting cuts generated by L-shaped partitioning of $R$ together take $O(k \lg{n})$ worst-case time.
\end{proof}

\section{Handling updates in a dynamic orthogonal polygon with holes}
\label{sect:holes}

In this section, we devise an algorithm to update the set of vertex guards in guarding the free space of a dynamic orthogonal polygon with dynamic orthogonal holes. 
For any current orthogonal polygonal domain ${\cal P}'$, apart from type-I and type-II updates mentioned in Subsection~\ref{subsect:charaffreg}, we allow a {\it type-III update} in which an orthogonal polygon $R'$ is inserted to the interior of ${\cal P}'$ such that $R'$ does not intersect any hole of ${\cal P}'$ as well as the outer boundary of ${\cal P}'$.
\begin{wrapfigure}{r}{0.4\textwidth}
\centering
    \begin{minipage}[t]{\linewidth}
    \centering
    \includegraphics[height=2.75cm]{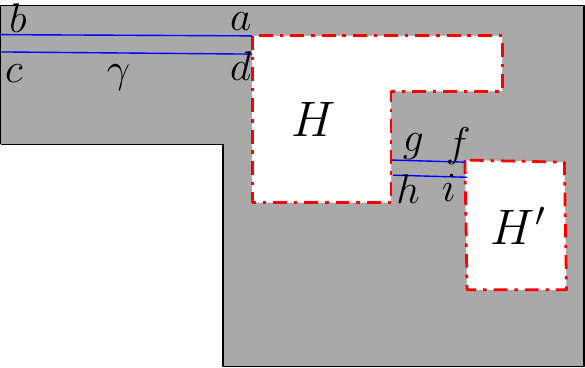}
    \end{minipage}
	\caption{Illustrating the channels of holes $H$ and $H'$ (blue) in an orthogonal polygonal domain $\cal P$.
	The channel of $H$ is incident to outer polygon of $\cal P$, whereas the channel of $H'$ is incident to $bd(H)$.
	The channel of $H$ (resp. $H'$) introduces three new vertices $b, c$, and $d$ (resp. $g, h, i$) into $P$, and $a$ (resp. $f$) is not a vertex of $P$.
	}
    \label{fig:connected_channel}
\end{wrapfigure}

Like in \cite{journals/dcg/BjorsachsSouv95}, by joining every hole to another hole or the outer boundary of the polygon, we reduce the problem of guarding the free space of an orthogonal polygonal domain to the problem of guarding a hole-free orthogonal polygon.
The joining is accomplished with thin horizontal rectangles, known as {\it channels}.
To vertex guard the resultant hole-free orthogonal polygon, we use the preprocessing and the guard update algorithms detailed in Section~\ref{sect:holefree}.
Unlike the channels constructed for polygonal domains in \cite{journals/dcg/BjorsachsSouv95}, the channels computed here are horizontal.
For any hole $H$ of an orthogonal polygonal domain $\cal T$, the {\it (horizontal) channel $\gamma_H$ of $H$} is a very thin rectangle in ${\cal F(T)}$:
let $a$ be the vertex with least $x$-coordinate value among all the vertices of $H$ with maximum $y$-coordinate value and let $b$ be the point of intersection of the leftward horizontal ray from $a$; then the the line segment between $a$ and $b$ is the top edge of the channel of $H$.
This way of defining $\gamma_H$ ensures that $\gamma_H$ does not strike $bd(H)$ before striking either another hole of $\cal T$ or the outer polygon of $\cal T$.
(Refer to Fig.~\ref{fig:connected_channel}.)
Since the channels are thin horizontal rectangles, no two channels intersect.
All other properties of channels remain same as in \cite{journals/dcg/BjorsachsSouv95}.
In particular, the following variant of a theorem from \cite{journals/dcg/BjorsachsSouv95} is useful.

\begin{propo}
\label{propro:holeless}
Any orthogonal polygonal domain $\cal T$ with $h'$ number of holes can be converted to a hole-free orthogonal polygon by removing $h'$ number of horizontal channels from $\cal T$.
\end{propo}

For any vertex $a$ of a hole $H$ of $\cal T$ with the channel $\gamma_H$ of $H$ incident to $a$, since $a$ is not a vertex in $P$, $\gamma_H$ introduces three additional vertices.
Hence, every channel introduces three new vertices; however, none of them is reflex and they together help in removing a reflex vertex of a hole.
Since the L-shaped partitioning algorithm \cite{journals/jourgeom/ORourke83} places guards at only reflex vertices, no guard is placed at any of the additional vertices introduced by channels.
Let $T$ be the hole-free orthogonal polygon obtained from $\cal T$ as described above.
Then, $T$ has a total of $(n'+2h')$ vertices, where $n'$ is the number of vertices of $\cal T$ and $h'$ is the number of holes of $\cal T$.  
(Refer to Fig.~\ref{fig:connected_channel}.)
Hence, in guarding $T$, the algorithm is allowed to use at most $\lfloor (n' + 2h')/4 \rfloor$ vertex guards.
The following proposition formalizes this observation.

\begin{propo}[Theorem~3.1, \cite{journals/dcg/BjorsachsSouv95}]
\label{propo:guardingholeless}
For any orthogonal polygonal domain $\cal T$ and the hole-free orthogonal polygon $T$ constructed from $\cal T$ as described, the guard placement of $T$ due to the L-shaped partitioning algorithm in \cite{journals/jourgeom/ORourke83} guards the free space of $\cal T$ using at most $\lfloor (n' + 2h')/4 \rfloor$ vertex guards.
Here, $n'$ is the number of vertices of $\cal T$ and $h'$ is the number of holes of $\cal T$.
\end{propo}

\subsection{Preprocessing}
\label{subsect:preprocholes}

\ignore {
\cite{conf/soda/GiyoKapl07}: space $O(n\frac{\lg{n}}{\lg\lg{n}})$, any insertion/deletion of a vertical edge $O(\lg{n})$ amortized time, query $O(\lg{n})$ worst-case time --- for \cal V

\cite{conf/soda/Morten03}: space $O(n\frac{\lg{n}}{\lg\lg{n}})$, any insertion/deletion of a channel $O(\lg{n})$ worst-case time, and the query in $O(k+\lg{n})$ worst-case time where $k$ is the size of the output --- for \cal H
}
We first describe a data structure from Gioya and Kaplan~\cite{conf/soda/GiyoKapl07}.
Given a horizontal ray $r$, in contrast to \cite{conf/soda/Morten03}, this data structure determines only the first line segment intersected by $r$ among all the line segments currently stored in the data structure.
This is done in $O(\lg{q'})$ worst-case time, where $q'$ is the number of vertical line segments present in the data structure at the time of the operation.
This data structure stores any $q'$ number of vertical line segments using $O(q'\frac{\lg{q'}}{\lg\lg{q'}})$ space; and, it takes $O(\lg{q'})$ amortized time to insert any vertical line segment into data structure or to delete any vertical line segment from the data structure.

Given an orthogonal polygonal domain $\cal Q$ with $q$ vertices and $h'$ (orthogonal) holes, using the algorithm to compute the data structure in \cite{conf/soda/GiyoKapl07}, in $O(q\lg{q})$ amortized time, we build a data structure $\cal V$ of size $O(q\frac{\lg{q}}{\lg{\lg{q}}})$ comprising all the vertical edges of $\cal Q$.
We use $\cal V$ to compute channels for all the holes of $\cal Q$ in $O(q+h'\lg{q})$ time: this is accomplished with $h'$ number of ray-shooting queries (with horizontal rays) among the vertical edges of $\cal Q$ in $\cal V$.

Using the data structure in Mortensen~\cite{conf/soda/Morten03} (described in Subsection~\ref{subsect:sepaffreg}), in $O(h\lg{h})$ worst-case time, we build a data structure $\cal H$ of size $O(h\frac{\lg{h}}{\lg{\lg{h}}})$ comprising all the channels (horizontal line segments) in $\cal Q$.
This data structure helps in efficiently finding all the channels that intersect any vertical edge of $R'$.

We also compute the hole-free orthogonal polygon $Q$ corresponding to $\cal Q$.
In addition, the preprocessing algorithm in Section~\ref{sect:holefree} is applied to $Q$.

\subsection{Algorithm to update the set of vertex guards}
\label{subsect:updateholes}

After any update to ${\cal P}'$, using Proposition~\ref{propro:holeless}, we transform $\cal P$ to a hole-free orthogonal polygon $P$ by updating the set of channels in ${\cal P}'$.
Then, we use the algorithm from Section~\ref{sect:holefree} to update the set of vertex guards to guard $P$.
From Proposition~\ref{propo:guardingholeless}, we note that these set of vertex guards suffice to guard $\cal P$.

In the following, we devise an efficient algorithm to transform $\cal P$ to $P$ by introducing at most one new channel and updating a subset of existing channels in ${\cal P}'$.
When a hole $H$ is modified such that the channel of $H$ needs to be modified, we update the channel of $H$ by a ray-shooting query. 
In the following, we describe the other case (type-III update) in which an orthogonal hole $R'$ is inserted to ${\cal P}'$.
The polygon $R'$ may intersect some of the current channels in ${\cal P}'$.
(Refer to Fig.~\ref{fig:channel_hitting}.)
Let $\Gamma$ be the set of channels that intersect with $R'$.
We need to update every channel in $\Gamma$.
\begin{wrapfigure}{r}{0.4\textwidth}
\centering
    \begin{minipage}[t]{\linewidth}
    \centering
    \includegraphics[height=4cm]{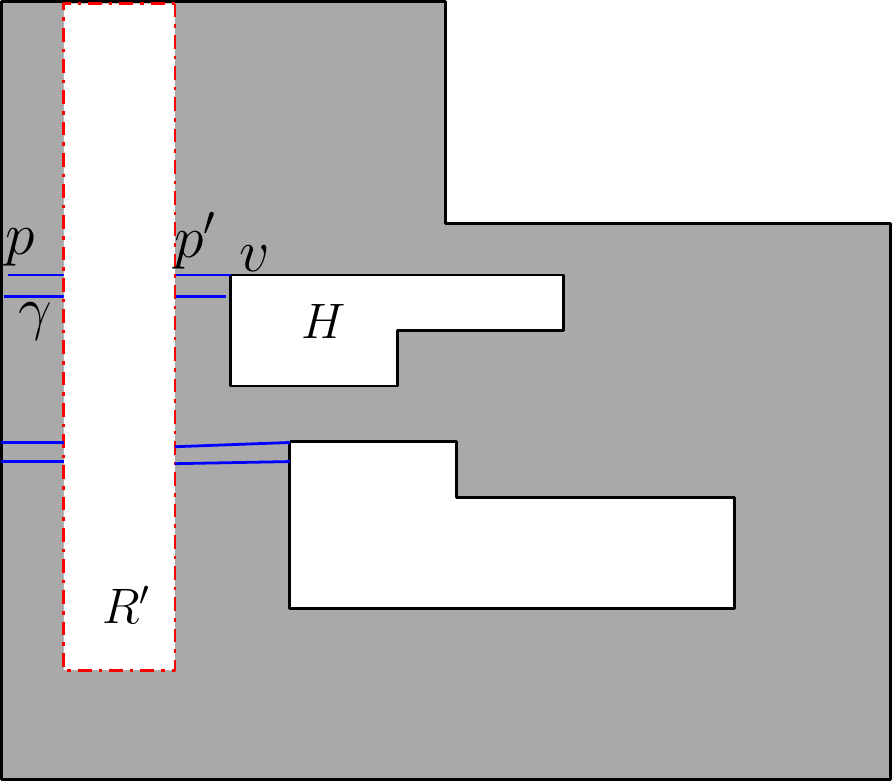}
    \caption{Illustrating an orthogonal polygon $R'$ intersecting channels in $\cal P'$.}
    \label{fig:channel_hitting}
    \end{minipage}
\end{wrapfigure}
Let $v$ be a vertex of a hole $H$ of $\cal P$ on which the channel of $H$ is incident.
For a channel $\gamma_H$ of a hole $H$ with endpoints $v$ and $p$, updating $\gamma_H$ involves determining the point $p'$ on line segment $vp$ such that $p'$ belongs to a vertical edge of $R'$ and the interior of line segment $vp'$ is in the free space of $\cal P$.
Essentially, the channel $vp$ is replaced with a horizontal channel with its top edge being $vp'$.
We update $\cal V$ data structure so that it comprises vertical edges of $\cal P$.
With the ray-shooting query with ray $vp$, using $\cal V$, we determine the point $p'$.
In addition, we update the channel $\gamma_H$ in the set $\cal H$ of channels.
Further, with a ray-shooting query, we determine the channel of $R'$.
There are only two ways in which any cut in $P'$ could get affected: due to its intersection with $R'$, and due to the reconstruction of affected channels in $\Gamma$.
Following the algorithm in Subsection~\ref{subsect:sepaffreg}, we determine all the affected cuts.
The following pseudocode gives the steps of the algorithm.

\begin{algorithm}
\caption{Update vertex guards of a dynamic orthogonal polygonal domain}
\label{code:guardpolydom}

Update the channels intersected by $R'$ using $\cal V$. \

Insert the modified channels into $\cal H$. \

Compute the channel of $R'$ using a ray-shoot. \

Update $\cal H$ to include the channel of $R'$. \

Update $\cal V$ to include the vertical edges of $R'$. \

Compute $bd(R)$ using the algorithm in Subsection~\ref{subsect:sepaffreg}. \

Determine the updated vertex guard set of the orthogonal polygon $P$ that corresponds to $\cal P$. \

\end{algorithm}

\begin{theorem}
\label{lem:orthopolyholecorr}
Given an orthogonal polygonal domain $\cal Q$ defined with $q$ vertices, we preprocess $\cal Q$ in $O(q \lg{q})$ time to construct data structures of size $O(q\frac{\lg{q}}{\lg\lg{q}})$ so that whenever any orthogonal polygonal domain ${\cal P}'$, which is obtained by a sequence of type-I, type-II and type-III updates to $\cal Q$, is updated to an orthogonal polygonal domain $\cal P$, the algorithm takes $O(k\lg{(n+n')})$ amortized time to guard $\cal P$ using at most $\lfloor (n+2h)/4 \rfloor$ vertex guards.
Here, $h$ is the number of orthogonal holes in $\cal P$, $n'$ and $n$ are the number of vertices of ${\cal P}'$ and $\cal P$ respectively, and $k$ is the sum of $|n-n'|$, the number of affected cuts in ${\cal P}'$ and the number of affected channels in ${\cal P}'$.
\end{theorem}
\begin{proof}
Subsection~\ref{subsect:preprocholes} details the time complexity of preprocessing algorithm and the space of data structures computed during that phase.
As in the proof of Theorem~\ref{thm:orthopolycorr}, let $k_1$ be the number of affected cuts in $P$ and let $k_2$ be the number of vertices of $R'$.
Also, let $k_3$ be the number of channels affected due to $R'$.
Computing the updated channels using $\cal V$ and inserting the modified channels into $\cal H$ together take $O(k_3\lg{n})$ amortized time.
Then, for type-I and type-II updates, computing the channels intersected by $R'$ and updating these channels together take $O((k_2 + k_3)\lg{n})$ amortized time.
In $O(k_3 \lg{n})$ worst-case time, we update $\cal H$ with the updated affected channels.
From Theorem~\ref{thm:orthopolycorr}, finding $bd(R)$, guarding $P$, and updating the associated data structures take $O((k_1+ k_2)\lg{n})$ worst-case time.
Updating $\cal V$ with vertical edges of $R'$ takes $O(k_2 \lg{n})$ amortized time.
Hence, after any type-I or type-II update, algorithm takes $O(k \lg{n})$ amortized time to update the set of vertex guards to guard $\cal P$, where $k$ is $k_1 + k_2 + k_3$.
For any type-III update, apart from the computations in type-I and type-II updates, we need to compute the channel due to $R'$.
Using $\cal V$, doing a ray-shooting query to find the channel of $R'$ takes $O(\lg{n})$ amortized time.
Like in type-I and type-II updates, updating $\cal V$ with vertical edges of $R'$ takes $O(k_2 \lg{n})$ amortized time.
Hence, any type-III update takes $O(k \lg{n})$ amortized time, where $k$ is $k_1 + k_2 + k_3$.
\end{proof}

\section{Conclusions}
\label{sect:conclu}
We presented an algorithm to update the set of vertex guards to survey the free space of a dynamic orthogonal polygonal domain.
It would be interesting to devise a dynamic algorithm for this problem, i.e., by removing the dependency on the non-input parameters from the update time complexity.
One possible extension of this problem is to consider guarding dynamic polygonal (not necessarily orthogonal) domains with point guards.
The other direction could be to maintain a set of guards so that the number of guards is an approximation to the optimal number of guards required to guard any dynamic polygonal domain.
Finally, upper bounding how far from optimal is the guard cover for the updated polygon in our algorithm could be worth exploring. 

\ignore {
1. The worst case bound of n/4 guards can be can be quite far from the optimal. Is there a way to (perhaps only allowing simpler updates than what is allowed here) maintain an (approximately) optimal vertex guard cover starting from an optimal cover of the original polygon?

2. Given that visibility is a global property of polygons (guards can see arbitrarily far way), how badly does an update in your algorithm destroy an optimal cover for P' as it is updated to P, i.e., how far from optimal is the resulting guard cover for P?

3. Given that n/4 is the upper bound for guarding also orthogonal polygons with holes; see the specific comment for lines 62-64, can you obtain this guarantee with your algorithm (perhaps modified)? Since any orthogonal hole has at least four vertices, h<=n/4 so in terms of n your bound is (n+2h)/4<=3n/8. Can you show a better bound in terms of just n for your algorithm?
}

\subsection*{Acknowledgement}

This research of R. Inkulu is supported in part by SERB MATRICS grant MTR/2017/000474.

\medskip

\bibliographystyle{plain}
\bibliography{../ajar/results/bibs/weiregsp,../ajar/results/bibs/geomgraphs,../ajar/results/bibs/misc,../ajar/results/bibs/shortestpaths,../ajar/results/bibs/visibility,../ajar/results/bibs/nearneigh,../ajar/results/bibs/voronoi,../ajar/results/bibs/artgall}

\end{document}